\documentclass[11pt,letter]{article}

\usepackage{framed}
\usepackage{setspace}
\usepackage{fullpage}
\usepackage{xcolor}
\usepackage[top=1in, bottom=1in, left=1in, right=1in]{geometry}
\usepackage{booktabs} 

\usepackage[ruled]{algorithm2e} 

\usepackage{mathtools}
\usepackage{multirow}
\usepackage{amsbsy, amsfonts, amsgen, amsmath, amsopn, amssymb, amstext,
	amsthm, amsxtra, bezier, bbm, color, enumerate, graphicx, latexsym, verbatim,
	pictexwd, supertabular, url, dsfont,leftidx, mathrsfs, mathtools, setspace, bbm}

\usepackage[ruled]{algorithm2e}

\newtheorem{theorem}{Theorem}
\newtheorem{lemma}[theorem]{Lemma}

\newcommand{\mb}[1]{\ensuremath{\boldsymbol{#1}}}

\newcommand{\onee}{\mathbbm{1}}

\newcommand{\eo}{\mathbb{E}_{\boldsymbol {\omega}}}

\newcommand{\bsigma}{\boldsymbol {\sigma}}

\def\opt{\textsc{OPT}}
\def\alg{\textsc{ALG}}
\def\galg{\textsc{G-ALG}}
\def\astalg{\textsc{Assort-ALG}}
\def\astgalg{\textsc{Assort-G-ALG}}
\def\dpg{\textsc{RBA}}
\DeclareMathOperator*{\argmax}{arg\,max}

\SetAlFnt{\small}
\SetAlCapFnt{\small}
\SetAlCapNameFnt{\small}
\SetAlCapHSkip{0pt}
\IncMargin{-\parindent}

\begin{document}
	\title {Online Allocation of Reusable Resources via Algorithms Guided by Fluid Approximations}
	\author {  Vineet Goyal \\ vgoyal@ieor.columbia.edu \and Garud Iyengar\\ garud@ieor.columbia.edu \and
Rajan Udwani \\ rudwani@berkeley.edu }
	\date{}
	\maketitle
\begin{abstract}
 We consider the problem of online allocation (matching and assortments) of reusable resources where customers arrive sequentially in an adversarial fashion and allocated resources are used or rented for a stochastic duration that is drawn independently from known distributions. Focusing on the case of large inventory, we give an algorithm that is $(1-1/e)$ competitive for general usage distributions. At the heart of our result is the notion of a relaxed online algorithm that is only subjected to fluid approximations of the stochastic elements in the problem. The output of this algorithm serves as a guide for the final algorithm. This leads to a principled approach for seamlessly addressing stochastic elements (such as reusability, customer choice, and combinations thereof) in online resource allocation problems, that may be useful more broadly.
\end{abstract}
\section{Introduction}
The problem of online matching introduced by Karp et al. \cite{kvv}, and its generalizations have played an important role in shaping the theory and practice of online algorithms. Traditionally, most of the work on online matching and more generally, online resource allocation, has focused on the case of resources that can only be used once \cite{kvv,pruhs,msvv,devhay,devsiv,survey,devanur,negin}. Recently, there has been a surge of interest in settings where each unit of a resource can be used multiple times \cite{dickerson, RST18, reuse, baek, feng,rba}. These settings present a generalization of classic online resource allocation and thus, retain all the trade offs incumbent in classical settings while adding new complexity due to the aspect of reusability. 

We are interested in finding online algorithms that do not have any advance knowledge (distributional or otherwise) of arrival sequence but compete well against the optimal offline/clairvoyant algorithm that has complete distributional knowledge and knows the entire arrival sequence in advance, but observes the stochastic realizations of usage durations in real-time when used resources return. In this direction, \cite{rba} recently proposed a policy called Ranking Based Allocation or \dpg\ for short, and showed that this achieves the best possible competitive ratio guarantee of $(1-1/e)$, for large starting inventory and when the usage distributions satisfy the IFR property (roughly speaking). For general distributions, the best known result is $1/2$ competitiveness of the greedy algorithm due to \cite{reuse}\footnote{There is a wealth of related past work on these kinds of problems. An appropriate discussion of this is beyond the scope of this note and we defer to \cite{rba} for such a discussion.}.

Focusing on the large inventory regime, we 
propose a new algorithm that is $(1-1/e)$ competitive for arbitrary usage distributions. This is the best possible competitive ratio guarantee. A major ingredient in achieving this result is the use of a relaxed online algorithm as a guide for the final algorithm. The relaxed algorithm is not subject to the stochasticity due to reusability (and choice) directly but rather only to their fluid approximations. This allows us to fully unlock the distributional knowledge of usage times. Before discussing this in more detail we review the model more formally and introduce some notation next. This is followed by a formal discussion of the results. We start with the setting of online $b-$matching with reusability and then give a general transformation of the result from matching to assortments that may be of independent interest.
\subsection{Model, Notation, and Prior Algorithms}
Let $G=(I,T,E)$ denote a bipartite graph with vertex sets $I,T$. The offline vertices or \emph{resources} will typically be addressed using index $i\in I$ and online vertices or \emph{arrivals} denoted using $t\in T$. We allow multiple copies or \emph{units} of each resource and let $c_i$ denote the number of units of resource $i$. This is also referred to as the starting inventory of resource $i$. Let $c_{\min}$ denote the minimum starting inventory over all resources. We are interested in the high inventory regime where $c_{\min}\to \infty$.

 We let the online vertices arrive in order of their index i.e., vertex $t+1\in T$ arrives after vertex $t\in T$. Overloading notation, we let $T$ also denote the total number of arrivals. The time of arrival of vertices $t\in T$ is given by $a(t)$ and we denote the set of times as $A(T)=\{a(1),\cdots,a(T)\}$. The arrival process is continuous so the arrival times can take any value and w.l.o.g., we allow at most one arrival to occur at any moment in time. Up on arrival of a vertex $t$, we observe the edges incident on $t$, denoted by set $S_t$. We must make an immediate and irrevocable decision to match $t$ to an unmatched unit of some resource in $S_t$ or reject $t$. If we match $t$ to a unit of resource $i$, the unit is used for a random duration $d_{it}$ that is independently drawn from usage distribution $F_i$. The unit returns to the system at time $a(t)+d_{it}$ and is available for a re-match immediately on return. The usage distribution is known to the algorithm in advance but the exact duration of usage is revealed only when a unit returns for reuse. 
  We earn a reward $r_i(d_{it})$ corresponding to this use of a unit of $i$, where $r_i(\cdot)$ can be any non-decreasing non-negative reward function with finite expectation $r_i=E_{d_{it}\sim F_i}[r_i(d_{it})]$. As shown in \cite{rba}, it suffices to consider only the expected rewards $r_i$ in any algorithm. The overall objective is to maximize the total reward with no knowledge of the sequence of arrivals. We evaluate performance of online algorithms for this problem by comparing against optimal offline algorithm that knows the entire graph in advance. More specifically, the offline algorithm knows everything known to an online algorithm and matches vertices in $T$ in the same order as an online algorithm, making irrevocable decisions. However, it also knows in advance the entire set of vertices in $T$, the set of edges $E$, and the arrival times $a(t)$ for all $t\in T$. Note that the realized usage durations are revealed to offline only when used units return, same as an online algorithm. This benchmark is also commonly referred to as the clairvoyant benchmark. As shown in \cite{rba}, in the asymptotic regime the performance of clairvoyant coincides with a natural LP benchmark introduced by Dickerson et al. \cite{dickerson}.  Our goal is to find an online algorithm that earns a fraction $\alpha$ of the expected total reward earned by offline on every possible arrival sequence $T,A(T)$.

We use $\galg,\alg$, and $\opt$ to refer to the online algorithms we analyze and the optimal offline respectively. For convenience, we also use these labels to denote the expected total reward of the respective algorithms. We let $\opt_i$ denote the total expected reward from usage of resource $i$ alone. To address the randomness in usage durations, we use $\omega$ to represent a sample path over random durations in \opt\ as follows: For each resource $i$ and for every unit $k$ of $i$, construct a long list (say $T$ entries) of independent samples from distribution $F_i$. When any unit of a resource is matched, we draw the first unused sample from the corresponding list of samples. This collection of lists and samples uniquely determines a sample path $\omega$. The matching output by \opt\ is a function of the sample path $\omega$ and we denote the set of arrivals matched to resource $i$ on sample path $\omega$ using $O(\omega,i)$. 

\subsubsection*{Prior Algorithms}
We review three prior algorithms for this problem before describing our algorithm. Let us start with the greedy algorithm, which is arguably the simplest one for this class of problems. Given a set of available resources $S_t$ with an edge to a new arrival $t$, the greedy algorithm matches $t$ to,
\[ \argmax_{i\in S_t}r_i, \]
i.e., an available resource with the maximum expected reward. As shown in \cite{reuse}, this is $1/2$ competitive for the problem in general and further, there are instances where the algorithm actually does no better than $1/2$. In fact, this implies that the worst case performance of greedy is identical for both reusable and non-reusable resources. This establishes a baseline result and typically one would like to do better than na\"ive greedy. 

For the case of non-reusable resources, a scalable and intuitive algorithm that achieves this goal is \emph{Inventory Balancing} (IB) \cite{negin,msvv}, sometimes also known as Balance \cite{pruhs}. The main insight in IB is to also pay attention to the fraction of inventory remaining and protect resources that are running low on available units. IB achieves this by considering an inventory dependent reduced price for each resource $i\in I$, given by,
\[r_i \left(1-g(y_i(t))\right),\]
where $y_i(t)\in[0,1]$ is the fraction of inventory of $i$ available when $t$ arrives and $g(x)=e^{-x}$ is a trade-off function that leads to the optimal guarantee of $(1-1/e)$ for $c_{\min}\to \infty$, even in the presence of stochastic elements such as customer choice \cite{negin}. Feng et al. \cite{feng} showed that IB also achieves this same guarantee for the case of deterministic reusability i.e., when every allocated unit of a resource is used for the same deterministic/fixed time duration.

In general, for reusable resources it is not obvious if IB effectively captures the fact that resources are used multiple times. In particular, \cite{rba} showed that even for the simple case of two-point usage distributions, IB may not appropriately capture ``expected" remaining inventory due to units that will return in the future. Specifically, while IB may protect resources that are low on remaining inventory, due to reusability these resources might be replenished in the near future and may not require protection. To address this, \cite{rba} proposed the \dpg\ algorithm which tracks availability of individual units of a resource. Formally, in \dpg\ one rank orders units of a resource $i$ in some order that is fixed throughout. Let $k\in[c_i]$ index the units. When a unit of $i$ needs to be allocated, the available unit which has the highest index will be allocated. This unit is then unavailable until returned. Given this, \dpg\ computes the reduced price as,
\[r_i \left(1-g\left(\frac{z_i(t)}{c_i}\right)\right),\]
where $z_i(t)$ is the highest available unit of $i$ when $t$ arrives. Observe that for non-reusable resources this collapses to IB. Interestingly, for reusable resources this appears to capture the effect of reusability despite being completely oblivious to the usage distributions. While \dpg\ achieves the guarantee of $(1-1/e)$ for a class of distributions that roughly speaking includes IFR distributions, its performance for general distributions is unknown. The analysis is substantially more intricate and does not easily generalize to the case when we also have stochasticity due customer choice. 
\subsection{Our Contributions}
An important challenge with analyzing algorithms such as \dpg\ in presence of stochastic elements is due to the \emph{adaptivity} of the algorithm. In particular, the highest available unit $z_i(t)$ of resource $i$ at any arrival $t$, is in general, a function of the durations of \emph{all past allocations} and thus, exhibits interesting and non-trivial behaviour. 
This problem is further compounded when we have stochasticity due to both reusability and customer choice (see Example H.1 in Appendix H of \cite{rba}). 

A possible solution to this would be to consider \emph{non-adaptive} algorithms. This refers to a class of algorithms for sequential decision making that do not depend on the past history of the random realizations and therefore, exhibit a uniformity in actions over all sample paths. Such algorithms are often easier to analyze as they are not directly influenced by the vagaries of stochastic parts of the problem. Unfortunately, such algorithms can suffer from poorer performance in comparison to adaptive ones, especially for the case of adversarial arrivals (see for instance, Mehta et al.\ \cite{mehta}). Interestingly, we find that in the large inventory regime this gap disappears and we devise non-adaptive algorithms with the best guarantee achievable by any online algorithm. Further, the non-adaptivity allows us to cleanly and seamlessly deal with the combination of stochasticity due to reusability and choice. 
\begin{theorem}\label{main}
\alg\ achives a guarantee of $(1-1/e-\delta)$ with $\delta=O\left(\sqrt{\frac{\log c_{\min}}{c_{\min}}} \right)$ for online $b-$matching with reusable resources. In the presence of customer choice, the appropriately modified \astalg\ achieves the same performance guarantee. For $c{\min}\to \infty$, this converges to the best achievable guarantee of $(1-1/e)$.
\end{theorem}

That the guarantee can be achieved by a non-adaptive algorithm is perhaps particularly surprising given the evidence that even adaptive algorithms like IB do not seem to suitably account for reusability. Given the partial success of \dpg\ in this regard, one might expect that an algorithm that is even more sensitive to reusability would be the right candidate for a general result. Nonetheless, a tight result is achievable for the general case via a non-adaptive algorithm that combines the insight from \dpg\ with a fluid approximation of reusability to make it non-adaptive. 

In order to describe the algorithm for $b$-matching with reusability, we start by proposing a deterministic online algorithm, referred to as \galg, that is allowed more power than a standard online algorithm. First, it is not subject to the stochasticity in usage durations and instead, experiences reusability in a \emph{deterministic fluid form}. In this respect, \galg\ has more power than even \opt, which is subject to reusability in the true stochastic sense. Second (and as a consequence), \galg\ is allowed to fractionally match each arrival to many resources. Note that one might consider allowing offline to also make fractional matches. In fact, it can be shown that w.l.o.g., \opt\ is given by a dynamic program that always matches integrally and so allowing fractional matches alone does not give \galg\ any advantage over \opt.

\begin{algorithm}[H]
	\SetAlgoNoLine
	\textbf{Inputs:} Set $I$ of resources, capacities $c_i$, usage distributions $F_i$\;
		\textbf{Outputs:} Values $x_{it}\in[0,1]$\;
	Let $S = I$, $g(t)=e^{-t}$, and 
	initialize $z_i(k)=1$ for every $i\in I,k\in[c_i]$\;
	\For{every new arrival $t$}{
		For every $i\in I,k\in[c_i]$ and $t\geq 2$, update values 
		\[ 
		z_i(k)=z_i(k)+\sum_{\tau=1}^{t-1} \Big(F_i\big(a(t)-a(\tau)\big)-F_i\big(a(t-1)-a(\tau)\big)\Big)y_i(k,\tau)
		\text{ \tcp{Fluid\,\, update}}
		\]\\
		Let $S_t=\{i \mid (i,t)\in E\}$\;
		Initialize $\eta=0$, $y_i(k,t)=0$, and $x_{it}=0$ for all $i\in S_t,k\in [c_i]$\;
		\While{$\eta<1$}{
			\For{$i\in S_t$}{
				$ k^*(i)=\underset{k\in[c_i]}{\arg\max}\, \{k \cdot \onee(z_i(k)>0)\} \text{ \tcp{Highest\,\, available\,\, unit} } $\
			}
			Let $ i^*=\underset{i\in S_t }{\arg\max} \, r_i \Big(1-g\big(\frac{k^*(i)}{c_i}\big)\Big) \text {\tcp{Maximum\,\, reduced\,\, price}}$\
			
			\textbf{if } $k^*(i^*)=0$ \textbf{ then } exit while loop\;
			Let $y(k^*(i^*),t)=\min\{z_i(k^*(i^*)),1-\eta\}$\;
			Increase $x_{i^*t}$, $\eta$, and decrease $z_i(k^*(i^*))$, by $y_{i^*}(k^*(i^*),t)$ $ \text{ \tcp{Fractional\,\, match}}$\
	}}	
	\caption{\galg}
	\label{galg}
\end{algorithm}
The variable $z_i(k)$ keeps track of the fraction of unit $k$ of $i$ that is available. Variable $y_i(k,t)$ is the fraction of unit $k$ matched to arrival $t$. Observe that the fluid update equation takes a deterministic view of reusability whereby, if fraction $y_i(k,t)$ of unit $k$ is matched to $t$ then by time $a(t)+d$, exactly $F_i(d)$ fraction of this allocated fraction returns. The values $x_{it}$ represent the total fraction of $t$ matched to resource $i$. Note that every arrival is matched fractionally to the unit with the highest reduced price. In this sense, \galg\ performs its fractional matching by using the rank based allocation rule. Note also that each arrival may be matched to numerous units of different resources. The while loop runs as long as some fraction of an arrival can still be matched. Observe that at every new arrival, the algorithm runs for a finite number of steps as the highest available unit $k^*_i$ decreases by at least 1 for some resource $i$ in each iteration of the while loop. Thus, it takes at most $\sum_i c_i$ number of steps for the algorithm to fractionally match an arrival. 
We show that the expected performance of \galg\ is always at least $(1-1/e)-O(1/c_i)$ of \opt. This demonstrates that by letting online algorithms have some additional power, one can get an asymptotically optimal competitive ratio result. We refer to this algorithm as \galg\ since it will be used to \emph{guide} the actual online algorithm, which we denote simply as \alg\ and describe later. 

Clearly, due to the fact that \galg\ experiences only a fluid form of reusability and updates the availability of resources accordingly, we cannot directly use the output of this algorithm to match arrivals.  Nonetheless, observe that since we know the usage distributions in advance we can implement \galg\ online and use the output of this algorithm in the actual online algorithm, \alg. 
Specifically, at every arrival we first compute the fractional match made by \galg\ and then use this to make a match in \alg\ by means of independent randomized rounding. 

\begin{algorithm}[H]
	\SetAlgoNoLine
	\textbf{Inputs:} Set $I$ of resources, capacities $c_i$, values $\delta_i$, and input from \galg\;
	Let $S = I$ and 
	initialize inventories $y_i(0)=c_i$\;
	\For{every new arrival $t$}{
		For all $i\in I$, update inventories $y_i(t)$ for returned units\;
		Let $S_t=\{i \mid (i,t)\in E\}$\;
		Get inputs $x_{it}$ from \galg\ and 
		independently sample $u\in U[0,1]$\;
		\textbf{if } $\Big(u\,(1+\delta_i)\in(\sum_{j\in S_t; j\leq i-1} x_{jt},\sum_{j\in S_t; j\leq i} x_{jt}]$ and $y_i(t)>0\Big)$ \textbf{ then } Match $t$ to $i$\;
		\textbf{else } Reject $t$\; 
	}
	\caption{\alg}
	\label{alg}
\end{algorithm}

Our overall algorithm comprises of running and updating the states of both \galg\ and \alg\ at each arrival. As intended, the component that accounts for reusability, i.e., \galg, only experiences a fluid form of reusability and the overall algorithm is non-adaptive. To analyze the overall algorithm we use the path based certificate developed in \cite{rba,stochrew}. We next review this certification and subsequently prove the online $b-$matching part of Theorem \ref{main}. In Section \ref{sec:asst}, we visit the generalization to assortment and more formally discuss our framework for going from matchings to assortments.

\subsection{Review of Path-based Certificate}
Consider non-negative values $\theta_i$ for $i\in I$, path based values $\lambda_t(\omega)$ for $t\in T$, and constraints,
\begin{eqnarray}
	\sum_t\eo[\lambda_t(\omega)] +\sum_i \theta_i \leq \beta \alg,\label{cert1}\\
	\eo\Big[\sum_{t\in O(\omega,i)} \lambda_t (\omega) \Big] +\theta_i \geq \alpha_ir_i \opt_i\quad \forall i\in I.\label{cert2}
	\end{eqnarray}
Recall that $\omega$ denotes a sample path over stochastic usage durations in \opt. Existence of a feasible solution to the above system implies \alg\ is $\frac{\min_i \alpha_i}{\beta}$ competitive (Lemma 4 in \cite{rba}). Therefore, the above linear conditions in $\theta_i$ and $\lambda_t(\omega)$ act as a certificate. Finding a feasible solution for this set of conditions certifies a competitive ratio guarantee for \alg.  More simply, especially for online $b-$matching, we use path independent variables $\lambda_t$.

\section{Analysis for Online $b-$matching with Reusability}
The following process will be extremely important in the analysis. This is an explicitly defined random process that is independent of the complexities of \alg\ and \opt\ and a key piece of the analysis will involve an understanding of how the expected reward in this process responds to changes in the parameters of the process.

\textbf{Generalized $(F,\bsigma,\mb{p})$ random process:} We are given a single unit of a resource and ordered set of points $\bsigma=\{\sigma_1,\cdots,\sigma_{T}\}$ on the positive real line such that, $0<\sigma_1< \sigma_2< \cdots< \sigma_T$. These points are also referred to as arrivals and each point is associated with a probability given by the ordered set $\mb{p}=\{p_1,\cdots,p_T\}$. The resource is in one of two states at every point in time - \emph{free/available} or \emph{in-use/unavailable}. It starts at time $0$ on the line in the available state. The state of the unit evolves as we move along the positive direction on the line. If the unit is available just prior to a point in $\sigma_t\in \bsigma$, then with probability $p_t\in \mb{p}$ the unit independently becomes in-use for a random duration $\mb{d}$ drawn independently according to distribution $F$. The unit stays in-use from $(\sigma_t,\sigma_t+\mb{d})$ and switches back to being available at time $\sigma_t+\mb{d}$. Each time the unit switches from available to in-use we earn unit reward. Let $r(F,\bsigma,\mb{p})$ denote the total expected reward of the random process. When $\mb{p}$ is the set of all ones, we drop it from notation. Thus, $(F,\bsigma)$ represents the random process where probabilities associated with all arrivals in $\bsigma$ are unity. This is also consistent with the random process defined in previous work \cite{rba}. Another useful shorthand we use is $\mb{1}_{\bsigma}$ to denote a set of unit probabilities corresponding to arrivals in set $\bsigma$. For two sets $\mb{p}_1$ and $\mb{p}_2$ of probabilities over the same set of arrivals, we use $\mb{p}_1\vee \mb{p}_2$ to denote the probability set with maximum of the two probability values for each arrival.

When the probabilities $\mb{p}$ are all set to one, we recover the original random process defined in \cite{rba}. Now, consider a fluid version of this process.

\textbf{Fluid $(F,\bsigma,\mb{p})$ process:} We are given a single unit of a resource and ordered set of points/arrivals $\bsigma=\{\sigma_1,\cdots,\sigma_{T}\}$ on the positive real line such that, $0<\sigma_1< \sigma_2< \cdots< \sigma_T$. Each arrival is associated with a fraction given by the ordered set $\mb{p}=\{p_1,\cdots,p_T\}$. The resource is fractionally consumed at each point in $\bsigma$ according to $\mb{p}$. If a fraction $\delta_t$ of the resource is available when $\sigma_t\in \bsigma$ arrives, then $p_t\delta_t$ fraction of the resource is consumed by $\sigma_t$, generating reward $p_t\delta_t$. The $p_t\delta_t$ fraction consumed at $\sigma_t$ returns deterministically in the future according to the distribution $F$ i.e., exactly $F(d)$ fraction of $p_t\delta_t$ is available again by time $\sigma_t+d$, for every $d\geq 0$. 

\begin{lemma}\label{equiv}
	The probability of match to any arrival $t$ in the $(F,\bsigma,\mb{p})$ random process is the same as the fraction of resource available at arrival $t$ in the fluid counterpart. Consequently, the expected reward in every $(F,\bsigma,\mb{p})$ random process is exactly equal to the total reward in the fluid $(F,\bsigma,\mb{p})$ process.
\end{lemma}
\begin{proof}
	The proof hinges on the fact that in the random process, the durations and randomness in state transitions are independent of past randomness. We can therefore write a recursive equation for the probability of reward at every arrival in the random process. For every arrival in $\bsigma$, let $\eta(\sigma_t)$ denote the probability that the resource is available when $\sigma_t$ arrives. We have the following recursion for every arrival,
	\[ \eta(\sigma_t)=\eta(\sigma_{t-1})\big(1-p_{t-1}\big)+\sum_{\tau=1}^{t-1} \eta(\sigma_{\tau}) p_{\tau}\big(F(\sigma_t-\sigma_{\tau})-F(\sigma_{t-1}-\sigma_{\tau})\big), \]
	where $\eta(\sigma_1)=1$. By forward induction, it is easy to see that this set of equations has a unique solution. Now, observe that if we were to use $\eta(\sigma_t)$ to represent the fraction of resource available at $\sigma_t$ in the fluid process, then we would obtain the same recursive relation with the same starting condition of $\eta(\sigma_1)=1$. Thus, the probability of resource availability at $\sigma_t$ in the random process is exactly equal to the fraction of resource available at $\sigma_t$ in the fluid process. Therefore, the expected reward $p_t\eta(\sigma_t)$ from a match occurring at $\sigma_t$ in the random process equals the reward from consumption at $\sigma_t$ in the fluid process. 
\end{proof}

\begin{lemma}\label{galgvopt}
	   For every instance of the problem i.e., graph $G$, arrival times $A(T)$, and usage distributions $F_i$, we have,
	  \[\galg \geq \alpha \opt, \text{ with } \alpha=\Big(1-1/e-O\Big(\frac{1}{c_{\min}}\Big)\Big).\]
\end{lemma}
\begin{proof}
	We use the path-based certification with the following variable setting,
	\begin{eqnarray}
	\lambda_t&&=\sum_{i\in I} r_i\sum_{k\in[c_i]} y_i(k,t) \bigg(1-g\Big(\frac{k}{c_i}\Big)\bigg),\label{lambda}\\
	\theta_{it}&&=r_i\sum_{k\in[c_i]} y_i(k,t) g\Big(\frac{k}{c_i}\Big),\nonumber\\
	\theta_i&&= e^{\frac{1}{c_i}}\sum_{t} \theta_{it} \label{theta}.
	\end{eqnarray}
	Observe that the sum $\sum_t \lambda_t +\sum_i e^{-\frac{1}{c_i}}\theta_i$ is exactly the total revenue of $\galg$. Therefore, condition \eqref{cert1} is satisfied with $\beta=e^{\frac{1}{c_{\min}}}$. Fix an arbitrary resource $i$, we would like to show condition \eqref{cert2} i.e.,
	\[\eo\Big[\sum_{t\in O(\omega,i)} \lambda_t\Big] +\theta_i \geq \alpha_i r_i \opt_i, \]
	for $\alpha_i=1-1/e$. To show this we often need to refer to the state of \galg\ just after $t$ departs the system. In discussing this it is imperative to avoid any inconsistencies in boundary cases (such as when $F_i(0)=0$). Therefore, we formally consider the following sequence of events at any moment $a(t)$ when an arrival occurs: (i) Units that were in use and set to return at $a(t)$ are made available and inventories updated accordingly. (ii) Arrival $t$ is (fractionally) matched and state of matched units is adjusted by reducing inventory. Accordingly, we refer to the state of the system after step (i) is performed as the state at $t^-$ and the state after step $(ii)$ is performed as state at $t^+$.
	Correspondingly, for each arrival $t$ with an edge to $i$, and every unit $k$, we define indicator $\onee(\neg k,t^+)$ that takes value one if no fraction of $k$ is available at $t^+$. In other words, this indicates the unavailability of $k$ right after $t$ has been fractionally matched and state of units updated accordingly. 
	Consequently, for every $i\in S_t$ and $k\in[c_i]$ with $\onee(\neg k,t^+)=0$ i.e., some fraction of $k$ available at $t^+$, we have by definition of \galg\ and $\lambda_t$ in \eqref{lambda},
	\[\lambda_t \geq r_i\Big(1-g\Big(\frac{k}{c_i}\Big)\Big). \]
	Let $\Delta g(k)= g\big(\frac{k-1}{c_i}\big)-g\big(\frac{k}{c_i}\big)$. 
	Using this definition together with above observation, we decompose the contribution from $\lambda_t$ terms to \eqref{cert2} as follows,
	\begin{eqnarray}
	\eo\Big[\sum_{t\in O(\omega,i)} \lambda_t\Big]\geq r_i \eo\Big[\sum_{t\in O(\omega,i)} \Big(1-1/e-\sum_{k\in[c_i]} \Delta g (k) \onee(\neg k,t^+)\Big)\Big].\label{decompose}
	\end{eqnarray}
Based on this decomposition, we will now focus on proving that,
	\begin{eqnarray}
	 \eo\Big[\sum_{t\in O(\omega,i)}\sum_{k\in[c_i]} \Delta g (k) \onee(\neg k,t^+)\Big]\leq \frac{1}{r_i}\theta_i,\label{interim}
\end{eqnarray}
Let us see why this would prove inequality \eqref{cert2} for resource $i$. Substituting \eqref{interim} in \eqref{decompose} we get,
\begin{eqnarray*}
		\eo\Big[\sum_{t\in O(\omega,i)} \lambda_t\Big]&&\geq (1-1/e)r_i\opt_i-\theta_i,
\end{eqnarray*}
which proves \eqref{cert2} with $\alpha_i=1-1/e$. The remaining proof is dedicated to establishing \eqref{interim}. Similar to previous work (Lemma 12 in \cite{rba}), we use a coupling between \galg\ and \opt\ to transform \eqref{interim} into a statement about properties of $(F_i,\bsigma,\mb{p})$ random processes. 

Fix resource $i$, unit $k$ in \galg, and unit $k_O$ in \opt. Let $\mb{s}(k)$ denote the set of arrival times where $\onee(\neg k,t^+)=1$ in \galg. Arrival times in $\mb{s}(k)$ are ordered in ascending order. For convenience we refer to arrivals and arrival times interchangeably and often refer to arrivals in the set $\mb{s}(k)$. We claim that at every arrival  $t\in\mb{s}(k)$, either $k$ is (fractionally) matched to $t$ or there exists an arrival $\tau<t$ to which $k$ is (fractionally) matched and subsequently, $\onee(\neg k,t'^+)=1$ for all $t'\in(\tau,t]$. This follows from the fact that $z_i(k)$ decreases only if some part of $k$ is matched. Given this, let $\bsigma(k)$ denote the subset of arrivals $t\in \mb{s}(k)$ to which $k$ is fractionally matched. Therefore, $\bsigma(k)$ is exactly the set of arrivals where some nonzero fraction of $k$ is available and this fraction is completely matched to the arrival. 
From the coupling used in Lemma 12 in \cite{rba} and equivalent form in \cite{reuse}, we have,
\[ \eo\Big[\sum_{t\in O(\omega, k_O)} \onee(\neg k,t^+)\Big]\leq r(F_i,\mb{s}(k)),\]
where $O(\omega,k_O)$ denotes the set of arrivals matched to unit $k_O$ of $i$ in \opt. 

We would now like to compare $r(F_i,\mb{s}(k))$ with the expected reward of \galg\ from matching $k$. To make this connection we interpret the actions of \galg\ through a generalized random process. 
Consider the time ordered set $\mb{\Sigma}$ of all arrivals with an edge to $i$. Consider some arrival $t\in \mb{\Sigma}$ and associate a probability $p_t(k)$ with it as follows: If no fraction of $k$ is matched to $t$ we set $p_t(k)=0$. If some fraction $\eta_t(k)$ of $k$ is available for match to $t$ and a fraction $\gamma_t(k)\leq \eta_t(k)$ is actually matched to $t$, we set $p_t(k)=\frac{\gamma_t(k)}{\eta_t(k)}$. Note that this associates a probability of one with every arrival in $\bsigma(k)\subseteq \mb{\Sigma}$, which recall, is exactly the set of arrivals where some nonzero fraction of $k$ is available and fully matched in \galg. We use $\mb{p}(k)$ to denote the ordered set of probabilities corresponding to $\mb{\Sigma}$. Now, consider the $(F_i,\mb{\Sigma},\mb{p}(k))$ random process and its fluid version. By definition of $\mb{p}(k)$, the fluid $(F_i,\mb{\Sigma},\mb{p}(k))$ process corresponds exactly to the matching of unit $k$ in \galg. 
Applying Lemma \ref{equiv} we have that $r_i\times r(F_i,\mb{\Sigma},\mb{p}(k))$ is exactly equal to the the total reward generated from matching $k$ in \galg\ i.e., 
\[ r\big(F_i,\mb{\Sigma},\mb{p}(k)\big) = \sum_t y_i(k,t).  \]
Next, we claim that the fraction of unit $k$ available in \galg\ at arrivals in $\mb{s}(k)\backslash \bsigma(k)$ is zero. By lemma \ref{equiv} this implies that in the $(F_i,\mb{\Sigma},\mb{p}(k))$ process, probability of matching to any arrival $t\in\mb{s}(k)\backslash \bsigma(k)$ is zero. The claim essentially follows by definition; in \galg\ any arrival in $\mb{s}(k)$ where some non-zero fraction of $k$ is available to match must also be an arrival in $\bsigma(k)$. Thus, all arrivals in $\mb{s}(k)\backslash \bsigma(k)$ are such that no fraction of $k$ is available to match in \galg. 
Now, consider the augmented set of probabilities $\mb{p}(k)\vee \mb{1}_{\mb{s}(k)}$, which denotes that we set a probability of one for every arrival in $\mb{s}(k)$. Applying Lemma \ref{zeroset} we have that,
\[r\big(F_i,\mb{\Sigma},\mb{p}(k)\vee \mb{1}_{\mb{s}(k)}\big)=r\big(F_i, \mb{\Sigma},\mb{p}(k)\big).\]
Finally, using the monotonicity Lemma \ref{monotone} (proven subsequently), we have,
\[ r\big(F_i,\mb{s}(k)\big)\leq r\big(F_i,\mb{\Sigma},\mb{p}(k)\vee \mb{1}_{\mb{s}(k)}\big)= r\big(F_i,\mb{\Sigma},\mb{p}(k)\big),\] 
where we used the equivalent form $(F_i,\mb{\Sigma},\mb{0}\vee \mb{1}_{\mb{s}(k)})$ for the random process $(F_i,\mb{s}(k))$. The set $\mb{0}\vee \mb{1}_{\mb{s}(k)}$ represents a probability of zero corresponding to arrivals in $\mb{\Sigma}\backslash \mb{s}(k)$ and a probability of one for arrivals in $\mb{s}(k)$.
Simple algebra now completes the proof of \eqref{interim},
\begin{eqnarray*}
	 \eo\Big[\sum_{t\in O(\omega,i)}\sum_{k\in[c_i]} \Delta g (k) \onee(\neg k,t^+)\Big]&&\leq  \sum_{k_O \in [c_i]} \eo\Big[\sum_{t\in O(\omega,k_O)}\sum_{k\in[c_i]} g\Big(\frac{k}{c_i}\Big)(e^{\frac{1}{c_i}} -1) \onee(\neg k,t^+)\Big],\\
	 &&\leq \sum_{k_O \in [c_i]}\sum_{k \in [c_i]}g\Big(\frac{k}{c_i}\Big)(e^{\frac{1}{c_i}} -1) \, r\big(F_i,\mb{\Sigma},\mb{p}(k)\big),\\
	 &&\leq e^{\frac{1}{c_i}}\sum_{k \in [c_i]} \Big[g\Big(\frac{k}{c_i}\Big) \sum_t y_i(k,t)\Big],\\
	&& = \frac{1}{r_i} \theta_i.
\end{eqnarray*}
\end{proof}
\begin{lemma}\label{zeroset}
	Given a $(F,\bsigma,\mb{p})$ random process let $\bsigma' \subset \bsigma$ be some subset of arrivals such that, at each arrival in the set the probability of resource being in available state is zero. Associate the probability set $\mb{1}_{\bsigma'}$ of unit probabilities with arrivals in $\bsigma'$. Then, the random processes $(F,\bsigma,\mb{p})$ and $(F,\bsigma,\mb{p}\vee \mb{1}_{\bsigma'})$ are equivalent i.e., the availability of resource is the same in both processes at every arrival in $\bsigma$. 
\end{lemma}
\begin{proof}
	 	Observe that it suffices to show the lemma for a subset $\bsigma'$ consisting of a single arrival. The result for general $\bsigma'$ then follows by repeated application. So let $t$ denote an arbitrary arrival in $\bsigma$ such that in $(F,\bsigma,\mb{p})$ process the resource is in-use at $t$ w.p.\ 1. Observe that changing the probability associated with $t$ does not change the probability of resource being available at $t$ in the random process. Therefore, resource is in-use at $t$ w.p.\ 1, even in the $(F,\bsigma,\mb{p}\vee \mb{1}_t)$ random process where the probability associated with $t$ is one. 
\end{proof}
\begin{lemma}\label{monotone}
	For any distribution $F$, ordered set $\bsigma=\{\sigma_1,\cdots,\sigma_m\}$ of arrivals and associated probability sets $\mb{p}_1=\{p_{11},\cdots,p_{1m}\}$ and $\mb{p}_2=\{p_{21},\cdots,p_{2m}\}$ such that $p_{1t}\leq p_{2t}$ for every $t\in[m]$, we have,
	\[r(F_i,\bsigma,\mb{p}_1)\leq r(F_i,\bsigma,\mb{p}_2) .\]
\end{lemma}
\begin{proof}
	Sample arrival activations in advance and couple them so that any arrival activated due to in process $(F_i,\bsigma,\mb{p}_1)$ is also activated in process $(F_i,\bsigma,\mb{p}_2)$. Further, for every arrival $\sigma_t\in \mb{\sigma}$ that is not activated in process $(F_i,\bsigma,\mb{p}_1)$, allow the arrival to be active in process $(F_i,\bsigma,\mb{p}_2)$ w.p.\ $p_{2t}-p_{1t}$. So the overall set of activated arrivals in process 2 includes the set of arrivals in process 1. Now, the lemma follows from Lemma 13 in \cite{rba}.
\end{proof}
Now, recall the final online algorithm that experiences true stochasticity in usage, denoted \alg. Let $\delta_i=\sqrt{\frac{100\log c_i}{c_i}}$. When $t$ arrives, this algorithm samples a number $u\in U[0,1]$. Laying down the values $\frac{1}{1+\delta_i}x_{it}$ on the real line in arbitrary but fixed order over $i$, \alg\ matches $t$ to $i$ if $u$ falls in the relevant interval and $i$ is available. If $i$ is unavailable then $t$ is left unmatched.

\begin{lemma}\label{algvgalg}
For every resource $i$ with starting inventory $c_i\geq 3$, we have, $\alg_i \geq \frac{1-c^{-1}_i}{1+\delta_i} \galg_i$.
\end{lemma}
\begin{proof}
	The proof rests simply on showing that for every $(i,t)\in E$, $i$ is available at $t$ w.p. at least $1-c^{-1}_i$. This implies a lower bound of $\frac{1-c^{-1}_i}{1+\delta_i}x_{it}$ on the expected reward from matching $i$ to $t$, completing the proof. To prove the claim, consider an arbitrary edge $(i,t)\in E$ and let $\onee( i,t)$ indicate the event that some unit of $i$ is available in \alg\ to match to arrival $t$. Let $\onee(i\to t)$ indicate the event that random variable $u$ sampled at $t$ dictates $i$ be matched to $t$. Finally, let $\onee(d_t> a(\tau)-a(t))$ indicate that the duration of usage sampled for match at arrival $t$ is at least $a(\tau)-a(t)$. The event that a unit of $i$ is available when $t$ arrives is equivalent to the following event,
	\[ \sum_{\tau=1}^{t-1} \onee (i,\tau) \onee(i\to \tau)\onee(d_{\tau}> a(t)-a(\tau)) \leq c_i-1.\]
	The probability that this event occurs is lower bounded by the probability of the following event occurring,
	\[\sum_{\tau=1}^{t-1} \onee(i\to \tau)\onee(d_{\tau}> a(t)-a(\tau)) \leq c_i-\delta_i. \]
	Define Bernoulli random variables $X_{\tau}= \onee(i\to \tau)\onee(d_{\tau}> a(t)-a(\tau))$ for all $\tau\leq t-1$. Random variables $X_{\tau}$ are independent of each other as both $u$ and the duration of usage are independently sampled at each arrival. Further, the total expectation is upper bounded as follows, 
	\[\mu:=\mathbb{E}\Big[\sum_{\tau=1}^{t-1}X_{\tau}\Big]=\frac{1}{1+\delta_i} \sum_{\tau=1}^{t-1} x_{i\tau} \big(1-F_i( a(t)-a(\tau))\big)\leq\frac{c_i}{1+\delta_i}<c_i-2\sqrt{c_i \log c_i},\]
	where the first inequality follows from the definition of \galg. The second inequality follows from the fact that for $\delta_i= \sqrt{\frac{100\log c_i}{c_i}}$ we have, 
$\frac{1}{1+\delta_i}c_i< c_i-2\sqrt{c_i \log c_i}$ for $c_i\geq 1$. For $c_i\geq 3$ define non-negative quantity, $\eta=\frac{c_i-1}{\mu}-1$. Applying Chernoff we have for $c_i\geq 3$,
	\[ \mathbb{P}\Big( \sum_{\tau=1}^{t-1}X_{\tau} > (1+\eta)\mu = c_i-1\Big)\leq e^{-\frac{\mu\eta^2}{2+\eta}}<\frac{1}{c_i}. \]
	Therefore, we have that a unit of $i$ is available at $t$ w.p.\ at least $(1-1/c_i)$ for $c_i\geq 3$. 
\end{proof}
\begin{proof}[Proof of part of Theorem \ref{main}]
	From Lemma \ref{galgvopt} and Lemma \ref{algvgalg} we have that \alg\ is at least,
	\[(1-1/e)e^{-1/c_{\min}}\left(\frac{1-c^{-1}_{\min}}{1+\delta_{\min}}\right)\, \opt,\]
	competitive. For large $c_{\min}$ this converges to $(1-1/e)$ with convergence rate $O\left( \sqrt{\frac{\log c_{\min}}{c_{\min}}}\right)$.
\end{proof}
\section{From Matching to Assortments}\label{sec:asst}
In this more general setting, on arrival of a customer we see a choice model based on the type of the customer. Consequently, our goal is to offer an assortment of items to every arrival. The objective in this case is still to maximize the overall expected revenue, and we compare against optimal clairvoyant algorithms that know the choice models of all arrivals in advance but see the realizations of choice and usage durations in real-time (same as online algorithms). Further, in case of assortments we make the standard assumptions (\cite{negin,RST18,reuse}) that the set $\mathcal{F}$ of feasible assortments is downward-closed and the choice models are such that for any given $i$, the choice probabilities $\phi(i,S)$ of $i$ being chosen given set $S$, are monotonically non-increasing over nested collection of sets (larger sets, lower probability). Further, we also assume access to an oracle that given a set of prices and the choice model, outputs a revenue maximizing assortment (more generally, a constant factor approximation is also acceptable).

Since we now need to think in terms of sets of resources offered to arrivals, a relatively straightforward way to generalize \galg\ will be to fractionally ``match" every arrival to a collection of revenue maximizing assortments/sets, consuming constituent resources in a deterministic fractional fashion. 
In other words, we will convert the stochasticity due to choice into a fluid/deterministic version and develop a generalized version of \galg\ that achieves the desired competitive ratio w.r.t.\ \opt. 
Specifically, arrival $t$ will be fractionally matched to assortments $A(1,t),\cdots,A(m,t)$ for some $m\geq 0$, with fractions $y(j,t)>0$ such that $\sum_{j=1}^m y(j,t)\leq 1$. The amount of resource $i$ consumed as a result of this is given by $\sum_{A(j,t) \ni i} y(j,t)\, \phi\big(A(j,t),i\big)$. The collection of assortments is found by computing the revenue maximizing assortment with reduced prices computed as in the case of $b-$matching. We assume w.l.o.g.\ that the oracle that outputs revenue maximizing assortments never includes resources with zero probability of being chosen in the assortment. The fractions $y(j,t)$ are chosen to ensure that the inventory constraints are satisfied. 

\begin{algorithm}[H]
	\SetAlgoNoLine
	\textbf{Inputs:} Set $I$ of resources, capacities $c_i$, usage distributions $F_i$\;
	\textbf{Outputs:} For every arrival $t$, collection of assortments and probabilities $\{A(\eta,t),y(\eta,t)\}_{\eta}$\;
	Let $S = I$, $g(t)=e^{-t}$, and 
	initialize $z_i(k)=1$ for every $i\in I,k\in[c_i]$\;
	\For{every new arrival $t$}{
		For every $i\in I,k\in[c_i]$ and $t\geq 2$, update values 
		\[z_i(k)=z_i(k)+\sum_{\tau=1}^{t-1} \Big(F_i\big(a(t)-a(\tau)\big)-F_i\big(a(t-1)-a(\tau)\big)\Big)y_i(k,\tau)\]\\
		Let $S_t=\{i \mid (i,t)\in E\}$\;
		Initialize $\eta=0$, $y_i(k,t)=0$ for all $i\in S_t,k\in [c_i]$\; 
		\While{$\eta<1$}{
			\For{$i\in S_t$}{
				$k^*(i)=\underset{k\in[c_i]}{\arg\max}\, \{k \cdot \onee(z_i(k)>0)\} $\;
				\textbf{if} $k^*(i)=0$ \textbf{ then } $S_t=S_t\backslash\{i\}$\;
			}
			\textbf{if } $S_t=\emptyset$ \textbf{ then } exit while loop\;
			
			 Let $ A(\eta,t)=\underset{S\in S_t }{\arg\max} \, \sum_{i\in S} r_i \left(1-g\left(\frac{k^*(i)}{c_i}\right)\right)$ 
			 \tcp{Optimal\,\, assortment\,\, with\,\, reduced\,\, prices}\
			 
			 Let $y(\eta,t)= \min\left\{1-\eta, \min_{i\in A(\eta)} \frac{z_i(k^*(i))}{\phi(A(\eta,t),i)}\right\}$ \text{\tcp{Fluid\,\, update}}
			\ \\
			Increase $\eta$ by $y(\eta,t)$\;
			
			\For{$i\in A(\eta,t)$}{
				Increase $y_i(k^*(i),t)$ and decrease $z_i(k^*(i))$  by $y(\eta,t)\phi(A(\eta,t),i) $\;
			}
	}}	
	\caption{\astgalg}
	\label{astgalg}
\end{algorithm}
\begin{lemma}
	For every instance of the problem i.e., graph $G$, arrival times $A(T)$, and usage distributions $F_i$, we have,
	\[\galg \geq \alpha \opt, \text{ with } \alpha=\Big(1-1/e-O\Big(\frac{1}{c_{\min}}\Big)\Big).\]
\end{lemma}
\begin{proof}
	For a primer on analysing online assortments (with non-reusable resources) using the path-based certificate, we refer to Appendix H.1 of \cite{rba}. Let $\onee(\omega,i,t)$ denote successful allocation of a unit of resource $i$ to arrival $t$ in \opt\ on sample path $\omega$. Sample paths are now over the random usage durations as well as random choice of arrivals. Let $z_i(t)$ be the highest index unit of resource $i$ that has a non-zero fraction available at $t^+$. We use the path-based certification with the following variable setting,
	\begin{eqnarray}
	\lambda_t(\omega)&&= \sum_{i\in I} \onee(\omega,i,t) r_i \bigg(1-g\Big(\frac{z_i(t^+)}{c_i}\Big)\bigg),\nonumber
	\end{eqnarray}
$\theta_i$ is defined the same as in Lemma \ref{galgvopt}. The rest of the analysis now follows more or less verbatim from the case of $b-$matching.
\end{proof}
 The main new challenge will be to turn this into \alg. In doing so we must deal with scenarios where \astgalg\ directs some mass towards a set $A$ but only some subset of resources in $A$ are available in \alg. Recall that in case of matching if the randomly chosen resource is unavailable we simply leave $t$ unmatched. We could consider a similar approach here whereby if any unit of sampled set $A$ is unavailable then we do not offer $A$. However, this will not preserve the overall revenue in expectation as the probability of every resource in $A$ being available simultaneously can be small. 
 If it were acceptable to offer an assortment with items that are not available in \alg\ then we could also offer the set $A$ as is. The underlying assumption in such a case is that if the arrival chooses an unavailable item then we earn no reward and the arrival simply departs. However, in many applications it may not be practically acceptable to do this as a stock out event after choosing an item is undesirable from the point of view of customer retention. 

Another approach could be to offer the subset $S$ of $A$ that is available in \alg\ at $t$. However, this can affect the probability of resources $i\in S\cap A$ being chosen in non-trivial ways and thus, affect future availability of resources in a way that is challenging to control. In other words, the Chernoff bound approach to arguing expected revenue is large enough may not apply. Consequently, we need to find a way to display some subsets of $A$ such that the overall probability of offering any single resource is not any larger and at the same time, we do not rely on many resources being available simultaneously. The main novelty of our approach to tackle this problem will be to switch our perspective from sets of resources back to individual resources. Specifically, for each resource we find the overall probability that the resource is chosen by a given arrival and then use these probabilities as our guideline i.e., given the subset $S\subseteq A$ of resources that is available, we find a new collection of assortments so that for every available resource, the overall probability of the resource being chosen matches this probability in the original collection of assortments in \astgalg. The main idea here is a probability matching, made non trivial by the fact that we are restricted to choice probabilities given by the choice model. We show that there is an iterative polytime algorithm that can do this probability matching and compute the new collection of assortments so that the probability of each available resource in \alg\ being chosen by $t$ matches that in \astgalg.

\begin{algorithm}[H]
	\SetAlgoNoLine
	\textbf{Inputs:} Set $I$ of resources, capacities $c_i$, values $\delta_i$, and input from \astgalg\;
	Let $S = I$ and 
	initialize inventories $y_i(0)=c_i$\;
	\For{every new arrival $t$}{
		For all $i\in I$, update inventories $y_i(t)$ for returned units\;
		Let $S_t=\{i \mid (i,t)\in E,\, y_i(t)>0 \}$\;
	{\color{black}	Get collection of assortments $\{A(\eta,t),y(\eta,t)\}_{\eta}$ from \astgalg\;
		For each $\eta$, compute $\mathcal{A}(\eta),\mathcal{Y}(\eta)=\text{Probability Match } \Big(A(\eta,t)\cap S_t,\{
		\frac{1}{1+\delta_s}\phi(A,s)\}_{s\in A(\eta,t)\cap S_t }\Big)$\;
	 Randomly sample collection $\eta$ w.p.\ $y(\eta,t)$. With probability \ $1-\sum_{\eta} y(\eta,t)$, reject arrival $t$\; 
	 From chosen collection $\eta$, randomly sample asssortment $A_j(\eta)$  w.p.\ $y_j(\eta)$, for $j\leq |\mathcal{A}(\eta)|$\;
	 Offer sampled assortment to $t$, otherwise reject $t$ w.p.\ $1-\sum_{j=1}^{|\mathcal{A}(\eta)|}y_j(\eta)$\;
	}}
	\caption{\astalg}
	\label{astalg}
\end{algorithm}

 \begin{algorithm}[H]
	\SetAlgoNoLine
\textbf{Inputs:} set $S$, choice model $\phi(\cdot,\cdot)$, target probabilities $p_s\leq \phi(S,s)$ for $s\in S$\;
\textbf{Output:} Collection $\mathcal{A}=\{A_1,\cdots,A_m\}$ with weights $\mathcal{Y}=\{y_1,\cdots,y_m\}$\;
Let $m=|S|$\;
\For{$j=1$ to $m$}{
Compute values $\gamma_s=\frac{p_s}{\phi(S,s)}$ for all $s\in S$\;
Let $s^*=\arg\min_{s\in S} \gamma_s$\;
Define $A_j=S$ and $y_j=\gamma_{s^*}$\;

For every $s\in S\backslash \{s^*\}$, update $p_s=p_s-y_j\,\phi(S,s)$\;
Update $S=S\backslash \{s^*\}$\;
}

	\caption{Probability Match $(S,\{p_s\}_{s\in S})$}
	\label{pmatch}
\end{algorithm}
%
\begin{lemma}\label{probmatch}
		Given a (substitutable) choice model $\phi: 2^{N},N \to [0,1]$, an assortment $A\subseteq N$ belonging to a downward closed feasible set $\mathcal{F}$, a subset $S\subseteq A$ and target probabilities $p_s$ such that, $\phi(A,s) \geq p_s$ for every resource $s\in S$. There exists a collection $\mathcal{A}=\{A_,\cdots,A_m\}$ of $m=|S|$ assortments along with weights $\big(y_i\big) \in[0,1]^{m}$, such that the following properties are satisfied:
		\begin{enumerate}[(i)]
			\item For every $i\in[m]$, $A_i \subseteq S$ and thus, $A_i\in \mathcal{F}$.
			\item Sum of weights, $\sum_{i\in[m]} y_i \leq 1$.
			\item For every $s\in S$, $\sum_{A_i\ni s} \, y_i\, \phi(A_i,s) = p_s$.
		\end{enumerate}
	Further, there is an iterative $O(m^2)$ algorithm that finds such a collection $\mathcal{A}$ along with weights $(y_i)$.

\end{lemma}
\begin{proof}
We give a constructive proof that also outlines the polynomial time algorithm to find the collection $\mathcal{A}$. Let, 
\[q^0_s=\phi( S,s) \text{ and }\gamma^0_s=\frac{p_s}{q^0_s} \text{ for every $s\in S$.}\] 
Observe that $q^0_s\geq \phi(A,s)\geq p_s$, due to substitutability. Thus, $\gamma^0_s\leq 1$ for every $s\in S$.

Let $s_1$ be an element in $S$ with the smallest value $\gamma^0_{s_1}$. 
Let $A_1=S$ be the first set added to collection $\mathcal{A}$ with $y_1=\gamma^0_{s_1}$, so that $y_1 \phi(A_1,s_1)=p_{s_1}$. We will ensure that all subsequent sets added to $\mathcal{A}$ do not include the element $s_1$ and this will guarantee condition $(iii)$ for element $s_1$. 
Next, define the set $S^1=S\backslash \{s_1\}$. Let,
 \[q^1_s=\phi(S^1,s)\geq q^0_s\text{ and }\gamma^1_{s}=\frac{p_s-y_1q^0_s}{q^1_s} \text{  for every $s\in S^1$.}\]
  Observe that $\gamma^1_s\in[0,1]$ for every $s\in S^1$. Let $s_2$ denote the element with the smallest value $\gamma^1_{s_2}$, out of all elements in $S^1$. 
If $\gamma^1_{s_2}=0$ we stop, otherwise we now add the second set $A_2=S^1$ to the collection with $y_2=\gamma^1_{s_2}$. Inductively, after $i$ iterations of this process, we have added $i$ nested sets $A_i\subset A_{i-1}\subset \cdots\subset A_1$ to the collection and have the remaining set $S^{i}= A_{i}\backslash \{s_{i}\}$ of $|A|-i$ elements. Define values, 

\[q^{i}_s=\phi(S^{i},s)\text{ and }\gamma^{i}_s=\frac{p_s-\sum_{k=1}^{i}y_kq^{k-1}_{s}}{q^{i}_s} \text{ for every $s\in S^i$.}\]
 Let $s_{i+1}\in S^i$ be the element with the smallest value $\gamma^i_{s_{i+1}}$. If $\gamma^i_{s_{i+1}}>0$, we add the set $A_{i+1}=S^i$ to the collection with $y_{i+1}=\gamma^i_{s_{i+1}}$ and continue.

Clearly, this process terminates in at most $m=|S|$ steps, resulting in a collection of size at most $m$. Each step involves updating the set of remaining elements, computing the new values $\gamma^{(\cdot)}_s$ and finding the minimum of these values. Thus every iteration requires at most $O(m)$ time and the overall algorithm takes at most $O(m^2)$ time. Due to the nested nature of the sets and downward closedness of $\mathcal{F}$, condition $(i)$ is satisfied for every set added to the collection. It is easy to verify that condition $(iii)$ is satisfied for every element by induction. We established the base case for element $s_1$ in the first iteration. Suppose that the property holds for all elements $s_1,s_2,\cdots,s_{i-1}$. Then, by the definition of $y_{i}$ we have for element $s_i$,
\[\sum_{j=1}^{i} y_j \phi(A_j,s_i) =  \big(p_s- \sum_{j=1}^{i-1} y_j q^{j-1}_{s_i} \big)+\sum_{j=1}^{i-1} y_j q^{j-1}_{s_i}=p_{s_i}. \]
Since $s_i$ is excluded from all future sets added to the collection, this completes the induction for $(iii)$. Finally, to prove property $(ii)$ it suffices to show that,
\[y_{m}\leq \gamma^0_{s_{m-1}}-\sum_{i=1}^{m-1} y_i, \]
as this immediately implies, $\sum_{i\in [m]}y_i\leq \gamma^0_{s_{m}}\leq 1$.
The desired inequality follows by substituting $y_{m}$ and using the following facts: (i) $q^{j}_{s_m}$ is non-decreasing in $j$ due to substitutability, (ii) $y_i\geq 0$ for every $i\in[m]$ as we perform iteration $i$ only if $y_i=\gamma^{i-1}_{s_i}>0$. Therefore,
\[ y_{m}=\frac{p_{s_m}-\sum_{i=1}^{m-1}y_{i}q^{i-1}_{s_m}}{q^{m-1}_{s_m}}\leq \frac{p_{s_{m}}-q^{0}_{s_{m}}\sum_{i=1}^{m-1}y_i}{q^{0}_{s_{m}}} =\gamma^0_{s_{m}}-\sum_{i=1}^{m-1}y_i.\]
\end{proof}

In general, we have to probability match a collection of assortments $\mathcal{C}=\{C_1,\cdots,C_v\}$ with associated probabilities $p_v$, instead of a single assortment. This can be accomplished by performing the above process for each individual assortment $C_j, j\in[v]$ in the collection. Let $\mathcal{A}(j)=\{A_1(j),\cdots,A_{m_j}(j)\}$, denote the resulting collection of subsets of $C_j$ with associated weights $(y_{i}(j))$. The final collection is given by the union $\cup_{j=1}^v \mathcal{A}(j)$ with weights $\cup_{j=1}^v \big\{p_j y_{i}(j)\big\}_{i\in[m_j]}$. It is worth mentioning that one can perform this process somewhat faster by performing the iterative process directly for the menu of assortments $C$ rather than separately for each individual set $C_j$.  The algorithm can also be executed more efficiently for the commonly used MNL choice model. Due to the IIA property of MNL, we obtain that it suffices to sort the resources in order of values $\gamma^0_{s}$ in the beginning and this ordering does not change as we remove more and more elements. Each iteration only takes $O(1)$ time and so the process has runtime dominated by sorting a set of size $m=|S|$ i.e., $O(m\log m)$.

\begin{proof}[Proof of Theorem \ref{main}]
	To complete the proof of the main theorem we need to compare \astalg\ with \astgalg. Given the probability matching due to Lemma \ref{probmatch}, we have that for every arrival $t$, if a unit of resource $i$ is available in \astalg\ then it is offered to and chosen by arrival $t$ w.p.\ $\frac{1}{1+\delta_i}\sum_{k \in [c_i]}y_i(k,t)$. Observe that the total fraction of resource $i$ fluidly matched and chosen by arrival $t$ in \galg\ is exactly $\sum_{k \in [c_i]}y_i(k,t)$. Now the concentration argument used in proving Lemma \ref{algvgalg} gives us that $i$ is available at $t$ w.p.\ at least $1-1/c_i$, completing the proof. 
\end{proof}

\bibliographystyle{alpha}
\bibliography{./../bib}
\end{document}